\documentclass[11pt,a4paper,final]{article}
\usepackage[left=1in, right=1in, top=1in, bottom=1in]{geometry}
\usepackage{times}
\usepackage{lineno,hyperref}
\usepackage[utf8]{inputenc}
\usepackage[T1]{fontenc}
\usepackage[english]{babel}
\usepackage{amsmath}
\usepackage{amsfonts}
\usepackage{amssymb}
\usepackage{makeidx}
\usepackage{tfrupee}
\usepackage{tikz}
\usetikzlibrary{arrows.meta,positioning,fit,calc,shapes.geometric}
\usepackage{pgfplots}
\usepackage{enumitem}
\pgfplotsset{compat=1.17}
\usetikzlibrary{arrows.meta, positioning}
\usepackage{natbib}
\usepackage{graphicx}
\usepackage{ifvtex}
\usepackage{filecontents}
\usepackage{tikz}
\usepackage{subcaption}
\usetikzlibrary{patterns}
\usepackage{pgfplots}
\usepackage{float}
\usepackage{multirow}
\usepackage{longtable}
\pgfplotsset{compat=1.17}
\usetikzlibrary{backgrounds}
\usetikzlibrary{arrows}
\usetikzlibrary{shapes,shapes.geometric,shapes.misc}
\pgfdeclarelayer{edgelayer}
\pgfdeclarelayer{nodelayer}
\pgfsetlayers{background,edgelayer,nodelayer,main}
\tikzstyle{none}=[inner sep=0mm]
\tikzset{new style 0/.style={circle,draw}}
\newtheorem{theorem}{Theorem}
\newtheorem{corollary}{Corollary}
\newtheorem{lemma}{Lemma}
\newtheorem{proposition}{Proposition}
\newtheorem{assumption}{ASSUMPTION}
\newtheorem{definition}{Definition}
\newtheorem{remark}{Remark}
\newtheorem{testableprediction}{Testableprediction}
\newenvironment{proof}{\paragraph{Proof:}}{\hfill$\square$}
\title{Pairwise Beats All-at-Once: Behavioral Gains from Sequential Choice Presentation}
\author{Dipankar Das\\Assistant Professor, Goa Institute of Management\\Email ID:dipankar@gim.ac.in;dipankar3das@gmail.com
}
\begin{document}
\maketitle
\begin{abstract}
	This paper presents the Sequential Rationality Hypothesis, which argues that consumers are better able to make utility-maximizing decisions when products appear in sequential pairwise comparisons rather than in simultaneous multi-option displays. Although this involves higher cognitive costs than the all-at-once format, the current digital market, with its diverse products listed by review ratings, pricing, and paid products, often creates inconsistent choices. The present work shows that preparing the list sequentially supports more rational choice, as the consumer tries to minimize cognitive costs and may otherwise make an irrational decision. If the decision remains the same on both offers, then that is a consistent preference. The platform uses this approach by reducing cognitive costs while still providing the list in an all-at-once format rather than sequentially. To show how sequential exposure reduces cognitive overload and prevents context-dependent errors, we develop a bounded attention model and extend the monotonic attention rule of the random attention model to theorize the sequential rational hypothesis. Using a theoretical design with common consumer goods, we test these hypotheses. This theoretical model helps policymakers in digital market laws, behavioral economics, marketing, and digital platform design consider how choice architectures may improve consumer choices and encourage rational decision-making.
\end{abstract}
\pagebreak
\section{Introduction}
Consumers frequently encounter decision environments in digital markets characterized by high product diversity and digital interfaces where multiple options are presented simultaneously. On the other hand, classical economic models presuppose full rationality and stable preferences, which implies that agents can process and evaluate all alternatives concurrently to select the utility-maximizing option. Empirical evidence consistently reveals that human cognitive capacities are bounded. As established in foundational works by Simon \cite{simon1955behavioral}, Kahneman \cite{kahneman2003maps}, and Gabaix \cite{gabaix2014sparsity}, decision-makers often operate under informational and attention constraints, leading to deviations from full rationality.\\
Recent developments in the economics of recommender systems underscore the impact of choice architecture on consumer behavior in digital markets\cite{jannach2010recommender, del2016towards, de2024economic}. While recommender systems are designed to guide consumers through informational complexity, they often present alternatives simultaneously, which can lead to cognitive overload and a decline in decision quality. Against this backdrop, this paper introduces the \textit{Sequential Rationality Hypothesis}, which posits that a sequential, pairwise presentation of options can systematically enhance decision quality by aligning the structure of the choice environment with cognitive processing capabilities.\\
Using a simple numerical example, the \textit{Sequential Rationality Hypothesis} can be explained. Consider a hypothetical consumer choosing among three simple beverages: Coffee ($A$), Tea ($B$), and Juice ($D$). Each option has an associated utility reflecting the consumer's true preference:
 
 \begin{center}
 	\begin{tabular}{|c|c|}
 		\hline
 		\textbf{Product} & \textbf{True Utility} \\
 		\hline
 		Coffee ($A$) & $u(A) = 8$ \\
 		Tea ($B$) & $u(B) = 6$ \\
 		Juice ($D$) & $u(D) = 7$ \\
 		\hline
 	\end{tabular}
 \end{center}
 Let the consumer be shown all three products in the standard choice setting, viz. $\{A, B, D\}$, simultaneously, and is expected to select the most preferred option. Under the full attention model and assuming perfect rationality, the consumer would compare all three options and choose Coffee ($A$), the option that maximizes utility.\\
However, under limited attention, the consumer may only compare a subset, such as Tea ($B$) and Juice ($D$), and inadvertently overlook Coffee. If Juice is chosen based on this partial comparison, the decision is suboptimal:
\begin{equation}
	C(A,B,D) = D \quad \text{(irrational choice due to attention omission)}
\end{equation}
\indent  Now consider a sequential choice structure. In Step 1, the consumer has to choose between Coffee ($A$) and Tea ($B$). Since $u(A) > u(B)$, and the comparison is focused, Coffee is chosen:
 \begin{equation}
 	C(A,B) = A
 \end{equation}
 In Step 2, the consumer must choose between the previously selected Coffee ($A$) and Juice ($D$). Again, since $u(A) > u(D)$, Coffee is chosen:
 \begin{equation}
 	C(C(A,B), D) = A
 \end{equation}
 The final choice is Coffee ($A$), the utility-maximizing product.\\
\indent The Sequential Rationality Hypothesis is motivated by the theoretical literature on bounded rationality, list-based choice models, and limited attention frameworks \cite{cattaneo2020random,masatlioglu2012revealed,rubinstein2006model}. Whereas classical revealed preference theory assumes consideration of all alternatives, empirical research shows that attention is often confined to a subset, known as the consideration set \cite{hauser1990evaluation, goeree2008limited, van2010retrieving, honka2017advertising, reutskaja2011search}. Rubinstein and Salant \cite{rubinstein2006model} provide an axiomatic foundation for choice from lists, emphasizing how order effects shape the construction of preferences. However, such models remain static and do not capture dynamic, sequential attention shifts. Our hypothesis builds upon this gap by introducing a recursive model of binary comparisons under cognitive constraints, establishing conditions under which step-wise exposure enhances utility-maximizing outcomes.\\
\indent The central theoretical insight is that when consumers compare two options at a time, they avoid attention overload, pay cognitive costs, and are more likely to make consistent, rational decisions. Here the cognitive cost is the weighted sum of information process \cite{bettman1990componential,DRNSMSSJ}.This has been tested mathematically and has a logical formulation of the theory. This sequential offering acts as a cognitive filter that incrementally preserves the better alternatives, thereby reducing cognitive overload and enhancing decision-making. The sequential offering acts as a cognitive filter that incrementally preserves the better alternatives.
 Formally, when attention is limited then the following holds.
 \begin{equation}
 	\Pr(C(C(A,B),D) = \arg\max\{u(A), u(B), u(D)\}) > \Pr(C(A,B,D) = \arg\max\{u(A), u(B), u(D)\})
 \end{equation}
 The visual presentation of options in digital marketplaces has a significant influence on consumer decisions. Based on the following figure of Amazon's recommendation interface (Figure \ref{fig:srh-tree}), we explain user behavior through the lens of the \textit{Sequential Rationality Hypothesis} (SRH) using simple logic. The hypothesis proposes that sequential binary exposure mitigates attention overload and facilitates utility-maximizing decisions under bounded rationality, making us acutely aware of the influence of digital interfaces on consumer behavior.\\
The Amazon recommendation system presents a grid-like layout where alternatives are displayed both horizontally across rows and vertically across scrollable columns. This grid includes advertisements, promotional labels (e.g., "Deal of the Day"), brand tags (e.g., "Amazon's Choice"), and price discounts, all of which act as cognitive distractions unrelated to intrinsic utility.\\
From a behavioral point of view the following are the beatification may be emerged (i) horizontally aligned options force users to compare several visually similar products simultaneously, (ii) vertically stacked rows necessitate scrolling, introducing temporal fragmentation of attention,(iii) salient design cues (e.g., discount tags, star ratings, number of reviews) divert focus from utility-relevant attributes, (iv) the lack of structured comparison logic increases decision fatigue and lowers the probability of selecting the utility-maximizing option.\\
These factors align with the prediction of bounded attention models, where users fail to process all alternatives fully and are instead influenced by framing and salience \cite{kahneman2003maps, gabaix2014sparsity, bordalo2013salience}.\\
We develop a choice-theoretic model in which bounded attention and cognitive viability thresholds determine the selection path. Through a set of theorems, lemmas, corollaries, and propositions with illustrative counterexamples, we show, for example, that $C(C(A, B), D) \neq C(A, B, D)$ when attention is limited. This divergence departs from standard rational choice theory and provides a theoretical justification for pairwise sequential choice.\\
This research bridges the gap between cognitive limitations and economic rationality by integrating behavioral theory, formal modeling, and experimental design. For digital platforms and policymakers, our results suggest that optimal design of recommendation flows and disclosures should accommodate bounded rationality by staging comparisons sequentially. This could lead to more effective recommendation systems and better consumer decisions. For economic theorists, the Sequential Rationality Hypothesis advances a middle ground between full rationality and unstructured heuristics by embedding cognitive realism into decision theory. From the platform's point of view, the strategy of the choice architecture is to reduce the cognitive costs and make the lists like all at once format or simultaneously, so that the consumer would tend to buy the dedicated or targeted products. This is possible because there are no well-defined digital market laws.\\ \indent The article is structured as follows. It begins with a literature review, followed by a model section that explores sequential choice and probability of rationality outcomes, presenting the main theoretical results. This is then supported by a detailed discussion of the theoretical arguments for the sequential rationality hypothesis, a key aspect of the article. The section concludes with a summary of the conclusions, policy implications, and limitations.
 \begin{figure}[H]
 	\centering
 	\includegraphics[width=0.8\textwidth]{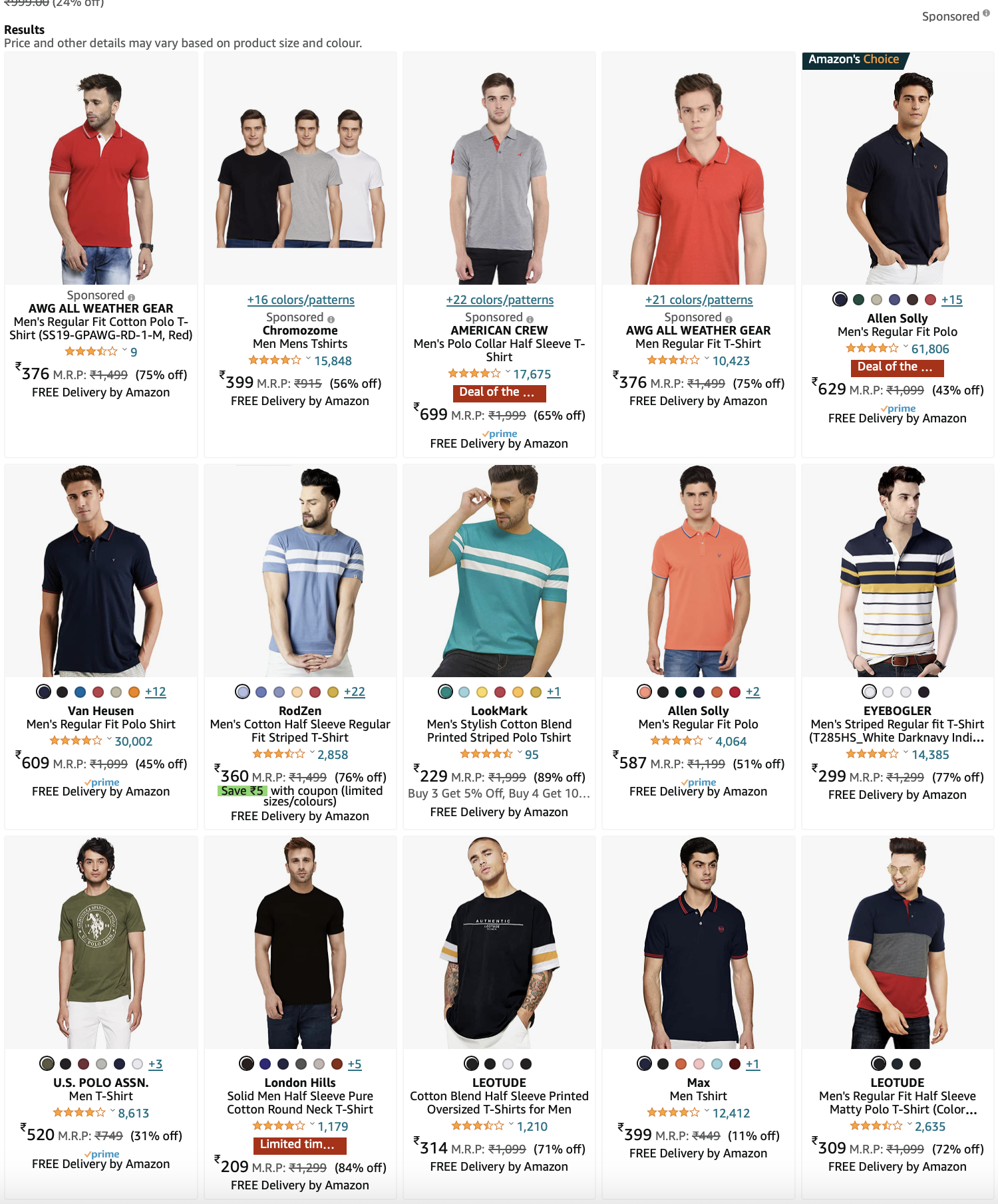}
 	\caption{Sequential Binary Comparisons under the SRH Model. Note: $\mu(\{A,C\}) = 0.4$ represents the probability of choosing $A$ over $C$ in a binary setting.}
 	\label{fig:srh-tree}
 \end{figure}

\section{Literature Review}
This section reviews theoretical and experimental contributions that inform the development of the Sequential Rationality Hypothesis. The literature is organized into two subsections: theoretical foundations and experimental insights. The motivation stems from the influence of choice architecture in digital markets, particularly where recommender systems shape consumer decision-making under cognitive limitations.\\
Recommender systems (RS), central to computational and economic decision-making environments, are designed to aid users in making decisions in information-rich settings \cite{jannach2010recommender}. A specific class, Economic Recommender Systems (ECRS), incorporates pricing and profit optimization into recommendation mechanisms \cite{de2024economic}, effectively transforming recommendations into a market-influencing economic tool. The central concern addressed in this paper is not the algorithmic design of RS per se, but how recommendations should be presented to consumers—sequentially (pairwise) or simultaneously—to optimize decision quality under bounded rationality.
\subsection{Theoretical Research}
Prior works, such as Simon \cite{simon1955behavioral}, Kahneman \cite{kahneman2003maps}, and Gabaix \cite{gabaix2014sparsity}, recognize bounded rationality and sparse information processing. Meanwhile, Tversky and Simonson \cite{tversky1993context} demonstrate that context and local comparisons can lead to violations of value maximization.
However, these models do not formally capture how attention-driven, sequential binary comparisons evolve or how preference reversals emerge endogenously from structural decision trees; the sequential rational hypothesis model fills this gap by introducing a rational but non-simultaneous decision logic under cognitive constraints.\\
Choice involves two types of uncertainty: uncertainty about future consequences of current actions and uncertainty about future preferences regarding those consequences \cite{march1978bounded, savage1972foundations}. Our work is concerned with the second uncertainty. If the list of alternatives is offered sequentially, pairwise, then this kind of uncertainty can be minimized.   In reference to the model of choice from lists, it can be said that an agent does not have a comprehensive set of elements in front; instead, the elements come to the mind in some sequential manner \cite{rubinstein2006model}. The introduction builds directly on Rubinstein and Salant's \cite{rubinstein2006model} foundational model of choice from lists by extending the core idea that the order of presentation influences decision-making. While their paper characterizes how list-based presentation affects rationality under axioms like Partition Independence, the current paper proposes the Sequential Rationality Hypothesis, offering a formal and experimental framework to test whether pairwise sequential offerings enhance utility-maximizing behavior by mitigating bounded attention constraints. While Rubinstein and Salant provide an axiomatic and representational framework showing that choice behavior can be influenced by the structure of lists (violating standard rationality axioms), they stop short of modeling dynamic or sequential attention processes. The present introduction addresses this gap by introducing the Sequential Rationality Hypothesis and proposing a formal model and experimental design to test how reducing cognitive load through step-wise comparisons may improve decision outcomes. Thus, the gap lies in translating the theoretical insights on list-based irrationality into experimentally testable predictions about attention-aware sequential choice. Revealed Preference (RP) theory assumes that the decision-maker selects the best available option after thoroughly considering all possible alternatives. A set of empirical research shows that RP theory is not always compatible with observed choices; instead, evidence suggests that the decision maker's choice is limited by their individual attention span \cite{hauser1990evaluation, goeree2008limited, van2010retrieving, honka2017advertising, reutskaja2011search, masatlioglu2012revealed}. Classical revealed preference theory is based on full attention, but the alternative theory, called the limited attention model, assumes that decision makers are assumed to
select the best available option from a subset of all possible alternatives, known as the consideration set \cite{masatlioglu2012revealed, lleras2017more, dean2017limited,masatlioglu2013choice}. 
Classical models such as McFadden's \cite{tversky1972elimination} random utility framework and subsequent revealed preference approaches (e.g., Manzini \& Mariotti, \cite{iyengar2000choice,caplin2015revealed}) assume that choice is made from full menus with simultaneous visibility of options, often using rational inattention or stochastic preference mechanisms.
The sequential rationality hypothesis model identifies a gap by proposing that agents use attention-constrained binary comparisons in a recursive sequence, where decision paths—not just outcomes—are shaped by cognitive effort and subset exposure, a mechanism unaccounted for in simultaneous or menu-dependent models.
The decision maker pays attention to the same subset of options every time she is confronted with the same set of available alternatives. A Random Attention Model (RAM) has been proposed that abstains from any specific parametric (stochastic) consideration set formation, and instead considers
a large class of nonparametric random attention rules \cite{cattaneo2020random}. While the paper \cite{masatlioglu2021decision} offers an insightful exploration of consumer decision-making within product networks—where each product (e.g., a movie) is a node and recommendations form links—it adopts a model where consumers traverse a graph based on exposure paths determined by these links. Critically, the study assumes that decision-makers hold stable and unaffected preferences and acknowledges that sparse network structures may prevent complete discovery of alternatives. However, it does not account for how bounded attention interacts with such structures or how consumers may sequentially compare only visible pairs due to cognitive constraints. In contrast, the proposed Sequential Rational Hypothesis explicitly models the decision process as a recursive series of binary comparisons, where attention is stochastically allocated over subsets, leading to path-dependent and partially informed outcomes. This framework not only explains how attention constraints distort the final choice but also allows for empirical falsifiability and welfare comparisons between sequential and simultaneous exposure formats. Hence, this paper addresses a critical gap in understanding how cognitive constraints and sequential attention influence choice in sparse, recommendation-driven environments, advancing beyond network traversal to model the bounded decision logic of the consumer structurally. The status quo bias, i.e., the tendency to choose the status quo alternative, is more prevalent in larger choice sets \cite{iyengar2000choice}. This has been demonstrated in another paper, which utilized the decision maker's attention on that alternative. This assumption naturally generates the prediction
that status quo bias will be more prevalent in larger choice sets in which attention is relatively more scarce \cite{masatlioglu2005rational,dean2017limited}. There is well-established evidence that decision makers consistently fail to consider
all available options. Instead, they restrict their attention to only a subset of alternatives and then undertake a more detailed analysis of this reduced set \cite{lleras2017more}. The decision-maker (DM) pays attention to the top N elements according to some ranking. Despite its simplicity and intuitive appeal, that the decision maker always selects the option that is highest in that order, experimental evidence indicates that this principle is often violated \cite{huber1982adding,simonson1992choice}. An article also finds that  the decision is done based on a context-dependent model of choice \cite{tversky1993context}. This means the decision maker rather than trying to find the maximal available alternative, tries only to find a satisfactory alternative, one that is ‘good enough \cite{caplin2011search,caplin2011search}. However in a research article related to the examination of how the timing of investment affects the levels of quality chosen by firms finds that aggregate profit is higher under sequential quality choice, suggesting that sequential quality choice offers coordination benefits, an insight often overlooked in quality choice models \cite{aoki1997sequential}. This is supported by the optimal search model proposed by \cite{dc603527-709e-3110-a352-6b21b157c70c} and further extended in \cite{ursu2025sequential} and concluded that the consumer searches sequentially and her goal is to maximize her expected utility net of total costs. On an average, decision maker stop searching too early and the tendency to search too little can be (partly) explained by learning processes \cite{SONNEMANS1998309,10.1093/restud/rdy023}. 
\subsection{Experimental Literature}
Sequential rationality builds on the foundational frameworks of economic rationality by imposing rationality at every stage of decision-making, not merely in outcomes. In classic decision theory, as formalized by \cite{savage1972foundations} and further operationalized by \cite{kreps1982sequential}, a strategy is considered sequentially rational if it maximizes the expected utility given every possible belief at each decision node. This principle has been widely applied in extensive-form games, where agents must optimize not only globally but also conditionally at every subgame.
Herbert Simon’s theory of bounded rationality \cite{simon1955behavioral} introduced cognitive constraints that limit agents’ abilities to pursue ideal optimization, prompting subsequent models, such as Prospect Theory \cite{kahneman1979d}, which integrate psychological biases into choice behavior. Kahneman \cite{kahneman2003maps} and Tversky \cite{tversky1993context} further demonstrated that context, heuristics, and framing effects systematically distort rational evaluation in sequential decision settings. Gabaix \cite{gabaix2014sparsity} formalized these distortions via sparsity-based models, while \cite{bordalo2013salience} incorporated salience effects where attention distorts evaluation order.\\
Chung and Ely \cite{chung2007foundations} revisited sequential rationality under dominant strategy mechanisms, arguing that even dominant strategies can be systematically misapplied when agents harbor bounded forward-looking beliefs. Bustow and Townsend \cite{busemeyer1993decision} proposed Decision Field Theory as a dynamic cognitive alternative, showing that beliefs and preferences evolve during the decision process, often deviating from the consistency assumed in sequential rationality. Ariely’s \cite{ariely2008predictably} experiments illustrated predictably irrational behaviors in sequential contexts like sunk cost fallacy and decoy effects, indicating a systematic violation of subgame perfection. \\
In support of the proposed sequential rational hypothesis a strong argument is that  the decision maker uses a costly comparison decision aid such as price and quality comparison websites and apps \cite{https://doi.org/10.1002/bdm.2181}. Nowadays, many consumer choice decisions can be characterized as “easy-click” checking decisions: situations in which consumers can use a costly comparison decision aid to guide their choice or even decide for them.This supports that the platform is generate the list of products in such a way so that it would be very difficult for the customer to compare. As a result they pay extra to avail the aid service.
\\ Now consider a set of related literature on some sub-themes. 
\subsection{Bounded rationality \& limited attention}
In many situations where agents make decisions under uncertainty, information acquisition is costly (involving pecuniary, time, or psychological costs); therefore, agents may rationally choose to remain imperfectly informed about the available options. This idea underlies the theory of rational inattention, a concept of significant importance in decision-making under uncertainty. It has become a key paradigm for modeling bounded rational behavior \cite{sims2003implications,fosgerau2020discrete}. A systematic literature has been studied in \cite{https://doi.org/10.1111/ijcs.12970}. 
\subsection{Limited consideration \& random attention}
Distributions of consideration set sizes are an important consideration in the final decision \cite{hauser1990evaluation}. The consideration set of a rational consumer will represent trade-offs between decision costs and the incremental benefits of choosing from a larger set of brands. If evaluating a brand reduces biases and uncertainty in perceived utility, the decision to evaluate a brand for inclusion in a consideration set differs from the decision to consider an evaluated brand. The decision to include a brand after evaluation entails a trade-off between the incremental benefits expected at each consumption occasion and the expected incremental decision costs associated with it. The decision to add to the consideration set is a sequential sampling strategy. Brands are evaluated in a specific order determined by a pre-evaluation. Consideration sets will be smaller for larger decision costs. Decision costs are additive. That is, the cost of evaluating $n$ brands is the sum of the decision costs for each brand. The proposed model in \cite{krajbich2011multialternative} supports this, assuming that the brain computes a relative decision value signal that evolves
over time as a Markov Gaussian process until a choice is made.
\subsection{Sequential vs simultaneous evaluation}
The typical consumer is also time-constrained and cannot afford to spend much time making each selection. To solve this decision problem, consumers need to perform a dynamic search over the set of feasible items under conditions of extreme time pressure and choice overload. The experimental results indicate that subjects can accelerate the search process when presented with a large number of items. The average number of items seen in a trial increases with set size. The initial search process is approximately random with respect to value; it is not the case that items with higher value are more likely to be seen \cite{reutskaja2011search}.
\subsection{Context effects that plague simultaneous menus}
Adding an alternative to a choice set can increase the probability of choosing the item that dominates it \cite{huber1982adding}. Regularity is the minimum condition of most choice models. It says that for any item that is a part of set A, where A is, in turn, a subset of B, the probability of choosing X from A must not be less than that of choosing X from B. If this inequality is satisfied, one cannot increase the probability of choosing an item by adding other items. This concept is employed in the present article to formalize the monotonicity axiom. The zero-comparison effect \cite{10.1086/657998} refers to a surprising finding that can deepen your understanding of consumer choices, showing how an option's attractiveness can be influenced by increasing an undesirable attribute from zero or a desirable one, even if it makes the option objectively less appealing. Recognizing this effect can inspire researchers and professionals to better interpret consumer decision-making processes.
Studies reveal that the zero-comparison effect only occurs with two options; when more options are available, consumers can make more meaningful attribute trade-offs, reducing reliance on relative perceptions and diminishing the effect \cite{https://doi.org/10.1002/cb.1889}. This limitation invites further exploration and understanding of consumer decision contexts.
\subsection{Choice overload \& assortment complexity}
Hick’s Law (or the Hick-Hyman Law) states that the more choices a person is presented with, the longer it will take the person to reach a decision \footnote{\url{https://www.interaction-design.org/literature/topics/hick-s-law?srsltid=AfmBOorpWedcNeBszR1gC3YNNhqxXp9W4Kaf5-pf569Ntam1GnZUlPE1}
}. On the other hand, having more choices is necessarily more intrinsically motivating than having fewer \cite{iyengar2000choice}. These two statements suggest that for a large set, the choice process should be sequential. A study
reveals that information overload negatively impacts perceived value and perceived effectiveness but positively impacts perceived risk \cite{https://doi.org/10.1002/cb.2098}. A paper  argues that the number of options under each label is more important in reducing choice overload than the number of labels. Authors call the number of options under each label “category ratio". Category ratio is represented as ()assortment size/category labels). The paper finds that a few labels are beneficial only when the category ratio is within the proposed optimal range. Uninformative labels also reduced choice overload when categorized using the optimal category ratio \cite{https://doi.org/10.1002/cb.2180}. A detailed systemic literature review has been done in. The paper finds that restricting options may no longer be an effective solution in this age of digitization and online purchasing \cite{https://doi.org/10.1111/ijcs.13029, https://doi.org/10.1111/ijcs.12899}. This supports our hypothesis that instead of restricting the options the options can be provided sequentially. 
\subsection{Applied Marketing Literature Supporting Sequential Rationality Hypothesis}
Due to limited processing capacity, consumers often lack well-defined existing preferences but construct them using various strategies contingent upon task demands. Consumer decision-making is dependent on constructive processing, which generally implies contingent choices \cite{JamesR}. This supports the sequential choice process. It has been proposed that, when choosing from a larger assortment, consumers with an available ideal point are likely to have stronger preferences for the chosen option than consumers without an available ideal attribute combination \cite{10.1086/376808}. Finding an ideal point is very difficult from a large assortment. Pairwise choice reduces this difficulty. The information economics literature suggests that search costs in electronic markets have essentially been reduced to zero, as consumers can use powerful search tools, free of charge, to easily find and compare product and shopping information on the Internet. In the present research, however, we propose a research model that suggests users must expend time and effort when completing search tasks, resulting in significant reduction of search costs and a trade-off between search cost and search performance \cite{8ec74f8b-6c4a-3cf7-8f86-a57f4d5c62ec}. This also supports that the rational choice is costly if offered a large assortment. When a consumer repeatedly purchases a product, they have become focused on that product \cite{moe2003buying}. The presence of search or evaluation costs may lead consumers not to search and not to choose if too
many or too few alternatives are offered \cite{3c7e062e-100d-3415-9cc5-a78e76a77992}. Some literature suggests that too many options can be detrimental to choice, creating a situation of choice overload and information overload \cite{CHERNEV2015333}. Research demonstrates that even when consumers make a purchase, the same item may generate lower satisfaction when chosen from a larger assortment, as opposed to a smaller one \cite{doi:10.1509/jmkr.47.2.312}. This suggests that the balanced approach is to offer pairwise. 
\subsection{Research Gap:} Despite the progress in modeling bounded rationality and attention, existing frameworks do not integrate sequential rationality with cognitively viable decision paths. The Sequential Rationality Hypothesis fills this gap by preserving game-theoretic consistency while embedding preference filtering and attention constraints into recursive decision structures. Unlike purely stochastic or network traversal models, SRH allows for falsifiable empirical tests and welfare comparisons in both digital and traditional choice environments.\\
The core contribution of this paper lies in proposing a new behavioral mechanism—\textit{sequential pairwise offering}—to improve rational choice under limited attention. While several theories in behavioral economics and decision sciences address context dependence, bounded rationality, and cognitive overload, the integration of a sequential decision structure with formal probabilistic choice theory remains relatively unexplored.
\section{Model}
We propose a model of "Sequential Rationality under Bounded Attention" grounded in the Random Attention Model (RAM) framework. Our objective is to formally capture the decision behavior of agents who select from a finite set of alternatives through sequential binary comparisons, rather than simultaneous multi-alternative choice, due to attention constraints. The model builds on a well-defined preference structure, a monotonic stochastic attention rule, and a recursive choice mechanism.
\subsection{Random Attention Theory and Sequential Rationality}
The sequential rationality framework presented earlier can be more rigorously grounded using the formalism of the \textit{Random Attention Model} (RAM) as introduced by \cite{cattaneo2020random}. RAM posits that consumers form \textit{stochastic consideration sets} from the set of available alternatives, and make utility-maximizing choices over the attended subset. This approach relaxes the classical assumption of full attention and accommodates empirically observed violations of regularity.\\
Let the decision maker (DM) have a strict preference ordering $\succ$ over the finite set of alternatives $X$, and is represented by a utility function $u:X \rightarrow \mathbb{R}$ such that $x \succ y$ if and only if $u(x) \ge u(y)$. It is assumed that the preference ordering is transitive and complete.\\
In the grand choice set $X$, each element represents a distinct product. A \textit{choice rule} is defined as a map $\pi: \mathcal{X} \times \mathcal{X} \to [0,1]$ such that $\pi(x\mid S)$ denotes the probability that product $x$ is chosen from set $S \subseteq X$, and $\sum_{x \in S} \pi(x\mid S) = 1$ for all non-empty $S$.\\
The stochastic choice behavior under RAM arises from a latent attention rule $\mu: \mathcal{X} \times \mathcal{X} \to [0,1]$, where $\mu(T\mid S)$ is the probability that the consideration set $T \subseteq S$ is formed when the choice set $S$ is presented. An important restriction imposed by RAM is monotonic attention as defined below. The following definitions have been rebuilt from (\cite{cattaneo2020random}).
\begin{definition}[Choice Rule]\label{1}
	Under the 	\textit{Choice rule} as defined as a map $\pi: \mathcal{X} \times \mathcal{X} \to [0,1]$ such that $\pi(x\mid S)$ denotes the probability that product $x$ is chosen from set $S \subseteq X$ with some positive probability in between $0 < \pi < 1$, and $\sum_{x \in S} \pi(x\mid S) = 1$ for all non-empty $S$. However, if the decision maker is selecting nothing then $\pi = 0$.
\end{definition}
\begin{definition}[Latent Attention Rule]\label{2}
	A latent attention rule $\mu: \mathcal{X} \times \mathcal{X} \to [0,1]$, where $\mu(T\mid S)$ is the probability that consideration set $T \subseteq S$ is formed when choice set $S$ is presented. 
\end{definition}
\begin{definition}[Monotonic Attention]\label{3} 
	Given the \textit{Choice Rule \& Latent Attention Rule} the monotone attention suggests
\begin{equation*}
\mu(T\mid S) \leq \mu(T\mid S \backslash  \{a\}) \quad \text{for all } a \in S \backslash T.
\end{equation*}
\end{definition}
This condition in Definition \ref{3}  ensures that removing a non-considered alternative from $S$ weakly increases the attention probability on $T$. This monotonic attention property is the key to supporting the sequential rationality and has been developed in the latter part of the article.
\begin{definition}[Cognitive Costs]
	The monotone attention definition suggests that when a decision maker selects a subset from an array of objects, cognitive costs play a role. When the subset $T$ is selected from $S$, it is expected that the cognitive cost will be less when selecting from $(T \ mid S \backslash \{a\})$. The rationality expects some amount of costs to be incurred. The platform, as a key player, can utilize this, for example, to design the market in a way that allows consumers to face lower cognitive costs. However, it may lead to irrational decisions when selecting an object. 
\end{definition}
\begin{definition}[Revised Choice Rule Under Monotone Attention]
	An agent with strict preference $\succ$ chooses the most preferred element out of any consideration set under potential consideration and hence realized/shortlist. Therefore,
	\begin{equation}
		\pi(x\mid S) = \sum_{T\subseteq S} \mathbf{1}\{x = \max_{\succ} T\} \cdot \mu(T\mid S),
	\end{equation}
	where $\mathbf{1}\{\cdot\}$ is the indicator function and $\max_{\succ}T$ is the $\succ$-maximal element of $T$. 
\end{definition}
\subsection{Attention Structure}
\subsubsection{Preliminaries}
Let $X$ be a finite set of alternatives and $\succ$ a strict preference relation on $X$. Let $\mathcal{C}$ denote the space of choice rules $C: \mathcal{P}(X) \rightarrow \Delta(X)$, where $\mathcal{P}(X)$ is the power set of $X$ excluding the empty set, and $\Delta(X)$ is the set of probability distributions over $X$.\\
Let $C^{\text{SEQ}}$ denote a choice function based on a binary tree structure: for $S \subseteq X$, $\mid S \mid> 2$, define the outcome recursively as:
\[
C^{\text{SEQ}}(S) = C(x, C^{\text{SEQ}}(S \setminus \{x\}))
\]
for any $x \in S$ (tie-breaking is notational and does not affect probabilistic RAM).\\
In line with RAM, the DM does not simultaneously evaluate all elements in a choice set $S \subseteq X$. Instead, the DM draws a subset $T \subseteq S$ according to an attention rule $\mu(T \mid S)$, representing the probability of considering $T$ when faced with $T$. We assume:
\begin{itemize}
	\item \textbf{Monotonicity:} For any $T\subseteq S \subseteq X$, and $a \in S \setminus T$, it holds that $\mu(T\mid S) \leq \mu(T\mid S \setminus \{a\})$.
	\item \textbf{Non-degeneracy:} For all nonempty $S$, $\sum_{T \subseteq S, T \neq \emptyset} \mu(T\mid S) = 1$.
\end{itemize}
The DM chooses the most preferred element from the attended set: $C(T) = \arg\max_{x \in T} u(x)$.
\subsection{Sequential Choice Process}
Given $S=\{x_1,x_2,\dots,x_n\}$ and a binary choice rule $C:\mathcal{X}\times\mathcal{X}\to\mathcal{X}$,
define $C^{\mathrm{seq}}$ recursively by
\[
\begin{aligned}
	&\text{Base case:} && C^{\mathrm{SEQ}}(x_i,x_{i+1}) \;=\; C(x_i,x_{i+1}).\\[4pt]
	&\text{Recursion:} && C^{\mathrm{SEQ}}(x_1,\dots,x_n)
	\;=\; C\!\big(x_1,\; C^{\mathrm{SEQ}}(x_2,\dots,x_n)\big), \quad n\ge 3.
\end{aligned}
\]
Or,
\[
\begin{aligned}
	& C^{\mathrm{SEQ}}(x_i,x_{i+1}) = C(x_i,x_{i+1}),\\
	& C^{\mathrm{SEQ}}(x_1,\dots,x_n)
	= C^{\mathrm{SEQ}}\big(C(x_1,x_2),x_3,\dots,x_n\big), \quad n\ge 3.
\end{aligned}
\]

This process forms a binary tree, and its output is a stochastic outcome driven by RAM's attention constraints applied over pairwise comparisons. This has been explained graphically in Figure \ref{fig:SCM} where the model has been compared with the simultaneous RAM as well.
\subsubsection{Comparison of Sequential and Simultaneous Choice Models}
\begin{figure}[H]
	\centering
	
	\begin{subfigure}[t]{0.45\textwidth}
		\centering
		\begin{tikzpicture}[
			sibling distance=25mm,
			level distance=15mm,
			every node/.style = {draw, rectangle, rounded corners,
				align=center, top color=white, bottom color=blue!10}
			]
			\node {Compare A vs Winner of (B vs C)}
			child {node {A}}
			child {
				node {Compare B vs C}
				child {node {B}}
				child {node {C}}
			};
		\end{tikzpicture}
		\caption{Sequential Choice Tree\\[1ex]
			$C^{\text{SEQ}}(A,B,C) = C(A, C(B,C))$\\[1ex]
			\textbf{Interpretation:} The DM first compares B vs C (C wins with probability greater than 0.50), then compares A vs C. Given A’s higher utility, A is more likely to win the final round. This structure emphasizes path dependence and simplifies cognitive load.}
				\label{fig:SCM}
	\end{subfigure}
	\hfill
	
	\begin{subfigure}[t]{1\textwidth}
		\centering
		\begin{tikzpicture}[
			node distance=4.0cm and 5.0cm, 
			every node/.style = {draw, circle, minimum size=14mm, font=\normalsize},
			every path/.style = {->, thick}
			]
			\node (S) at (0,8) {$\{A,B,C\}$};
			
			\node (T1) [below left=of S, xshift=-2cm, yshift=-0.5cm] {$\{A,B\}$};
			\node (T2) [below=of S, yshift=-0.5cm] {$\{A,C\}$};
			\node (T3) [below right=of S, xshift=2cm, yshift=-0.5cm] {$\{B,C\}$};
			
			\node (T4) [below=5.0cm of S, xshift=-6cm] {$\{A\}$};
			\node (T5) [below=5.0cm of S] {$\{B\}$};
			\node (T6) [below=5.0cm of S, xshift=6cm] {$\{C\}$};
			
			\draw (S) -- (T1) node[midway, left=6pt] {\small $\mu(\{A,B\})$};
			\draw (S) -- (T2) node[midway, right=6pt] {\small $\mu(\{A,C\})$};
			\draw (S) -- (T3) node[midway, right=6pt] {\small $\mu(\{B,C\})$};
			
			\draw[dashed] (S) -- (T4) node[midway, left=6pt] {\small $\mu(\{A\})$};
			\draw[dashed] (S) -- (T5) node[midway, right=6pt] {\small $\mu(\{B\})$};
			\draw[dashed] (S) -- (T6) node[midway, right=6pt] {\small $\mu(\{C\})$};
		\end{tikzpicture}
		
		\caption{
			Simultaneous RAM Choice:\\
			$\pi(x|\{A,B,C\}) = \sum_{T \subseteq \{A,B,C\}} \mu(T) \cdot \mathbf{1}_{x = \max(T)}$
		}
		\vspace{4pt}
		\small
		\textbf{Interpretation:} The DM may attend to any subset with probability $\mu(T)$. Higher-ranked options are chosen if attended, but the final choice may be suboptimal if top options are excluded. For example, if only $\{B,C\}$ is considered, C is chosen even though A is better.
	\end{subfigure}	
	\caption{Comparison of Sequential and Simultaneous Choice Models}
\end{figure}
\subsection{Axiomatic Characterization of Sequential Rationality under Bounded Attention}

To establish sequential rationality on firm theoretical ground, we now provide an axiomatic characterization. We define a choice rule as sequentially rational if it satisfies a specific set of behavioral axioms under bounded attention. These axioms ensure that the agent consistently applies transitive preferences over binary comparisons, with attention constraints embedded in the structure of choice.
\subsubsection{Axioms}

\begin{enumerate}[label=\textbf{(A\arabic*)}, leftmargin=2em]
	\item \textbf{Binary Consistency:} If $x \succ y$, then $C(\{x,y\}) = x$ with probability at least $p > 0.5$.
	
	\item \textbf{Sequential Transitivity:} For all $x, y, z \in X$, if $C(x,y) = x$ and $C(x,z) = x$, then $C(C(x,y),z) = x$.
	
	\item \textbf{Attention Monotonicity:} For any $T \subseteq S$ such that $x \in T$, the attention weight $\mu(T|S)$ satisfies:
	\[
	\mu(T|S) \leq \mu(T|S \setminus \{a\}) \quad \forall a \in S \setminus T.
	\]
	
	\item \textbf{Cognitive Parsimony:} For any $S, T \subseteq X$, if $|T| < |S|$ and $x \in T \cap S$, then:
	\[
	\pi(x|T) \geq \pi(x|S).
	\]
	This reflects the bounded attention cost of larger menus.
\end{enumerate}
\subsubsection{Bounded Rationality Operator }
We formalize a bounded rationality operator $C_\alpha(S)$ such that:

\begin{equation}
	\Pr(C_\alpha(S) = \arg\max_{x \in S} u(x)) = \Phi(|S|, \alpha)
\end{equation}

where $\alpha \in [0,1]$ represents attention fidelity, and $\Phi$ is a strictly decreasing function of $|S|$ for $\alpha < 1$. The decreasing nature is due to Axioms A3, and A4.
\section{Sequential Choice and Probability of Rational Outcomes}
We hypothesize that sequential choice structures can increase the probability of rational (utility-maximizing) outcomes by reducing cognitive load and improving attention focus. To formalize this, we employ a probabilistic attention model.\\
Let $p \in (0,1)$ denote the probability that a decision-maker (DM) successfully attends to any given item in a comparison. This parameter captures bounded cognitive capacity and is assumed constant per item.We define two choice architectures:
\begin{itemize}
	\item \textbf{Simultaneous Offering:} Consumer selects from $S = \{A, B, D\}$ directly.
	\item \textbf{Sequential Pairwise Offering/Sequential Offering:} Consumer first chooses from $\{A,B\}$, then compares the result with $D$.
\end{itemize}

We analytically compare the resulting choice probabilities under both settings, using the following core probability inequality:

\begin{equation}
	\Pr(C_\alpha(C_\alpha(A,B),D) = \arg\max\{u(A),u(B),u(D)\}) > \Pr(C_\alpha(A,B,D) = \arg\max\{u(A),u(B),u(D)\})
\end{equation}
\subsection{Simultaneous Offering}
Consider a choice set $\{A,B,D\}$ offered simultaneously. The DM must attend to all three items in a single evaluation. The probability of attending to all three is
\begin{equation}
	P_{\mathrm{SIM}} = p^3,
\end{equation}
which declines rapidly with set size when $p<1$, reflecting the cumulative attention burden.
\subsection{Sequential Offering}
Under a sequential architecture, the DM compares items pairwise in two stages:
\begin{enumerate}
	\item Step 1: Compare $\{A,B\}$.
	\item Step 2: Compare the winner of Step 1 with $D$.
\end{enumerate}
At each step, the DM must attend to only two items, with success probability
\begin{equation}
	P_{\mathrm{seq\mbox{-}step}} = p^2.
\end{equation}
Assuming independence across stages, the total probability of attending to all relevant items is
\begin{equation}
	P_{\mathrm{SEQ\mbox{-}total}} = (p^2)^2 = p^4.
\end{equation}
\begin{remark}
	Interpretation:Although $p^4 < p^3$ algebraically, this formulation does not incorporate the empirical advantage of reduced cognitive load in pairwise comparisons. Empirical evidence suggests that $p$ in a binary context is effectively higher than in a ternary context, making sequential architectures more accurate in practice. The advantages arise from:
	\begin{enumerate}
		\item Lower attention demands per decision point.
		\item Increased accuracy in binary comparisons.
		\item Reduced risk of attention misallocation.
	\end{enumerate}
	\end{remark}
	
	\paragraph{Setup.}
	Let $p_k\in(0,1]$ denote the \emph{effective per-item attention} (or reliability) parameter
	when $k$ items must be jointly processed. Empirically, reduced cognitive load in binary
	comparisons implies $p_2>p_3$. Here,$2\& 3$
	represent the number of objects from which the selection is to be made.
	\paragraph{Algebra with fixed $p$.}
	If one imposes a single $p$ for all arities (i.e. number of items), then
	\[
	P_{\mathrm{SIM}} = p^3,
	\qquad
	P_{\mathrm{SEQ}} = (p^2)^2 = p^4,
	\]
	and since $p\in(0,1)$, it follows that $p^4<p^3$, so the simultaneous architecture appears better.
	
	\paragraph{Heterogeneous $p$ by arity.}
	Allow the effective parameter to depend on arity:
	\[
	P_{\mathrm{SIM}} = p_3^{\,3},
	\qquad
	P_{\mathrm{SEQ}} = (p_2^{\,2})^2 = p_2^{\,4}.
	\]
	The sequential architecture is (weakly) more accurate iff
	\[
	p_2^{\,4} \;\ge\; p_3^{\,3}
	\quad\Longleftrightarrow\quad
	p_2 \;\ge\; p_3^{\,3/4}.
	\tag{$\star$}
	\]
	
	\begin{proposition}[Binary advantage threshold]
		If the effective binary parameter satisfies $p_2 \ge p_3^{3/4}$, then
		$P_{\mathrm{SEQ}} \ge P_{\mathrm{SIM}}$. In particular, any modest uplift
		$p_2/p_3 \ge p_3^{-1/4}$ suffices.
	\end{proposition}
	
	\begin{proof}
		Immediate from $P_{\mathrm{SEQ}}=p_2^4$ and $P_{\mathrm{SIM}}=p_3^3$.
		Taking logs preserves the inequality since both sides are positive:
		$4\log p_2 \ge 3\log p_3 \iff p_2 \ge p_3^{3/4}$.
	\end{proof}
	
	\begin{remark}[Interpreting the empirical claim]
		Condition $(\star)$ formalizes: although $p^4<p^3$ for a \emph{common} $p$,
		the empirically larger binary effectiveness ($p_2>p_3$ due to reduced cognitive load)
		can overturn the algebraic comparison. Even a modest increase in $p$ when moving from
		ternary to binary comparisons can make sequential architectures more accurate in practice. This is directly supported  by the axioms of attention monotonicity and cognitive parsimony (A3 and A4). 
	\end{remark}
	
	\paragraph{Illustration.}
	If $p_3=0.70$, then $p_3^{3/4}\approx 0.761$. Any binary effectiveness
	$p_2\ge 0.761$ yields $p_2^4 \ge p_3^3$, i.e., sequential accuracy surpasses simultaneous.
	
\subsection{Theoretical Results}
We now formalize the divergence between sequential and simultaneous outcomes under bounded attention.
\begin{assumption}[Attention and accuracy by arity]\label{assump:attn-acc}
	For each $k\ge 2$, let $\alpha_k\in(0,1]$ denote the per-item attention success probability when a $k$-ary set is evaluated in a single stage. Within a stage, attention to distinct items occurs independently. Let $\beta_k\in(0,1]$ denote the conditional probability that, given attention to all $k$ items, the chosen item is the true utility maximizer among those $k$. Across distinct stages, attention and choice events are independent.
\end{assumption}

\begin{definition}[Architectures and success]\label{def:success}
	Given a $k$-ary stage, define its \emph{stage success} as the event that (i) all required items in that stage are attended and (ii) the stage outputs the true utility maximizer among those items. For a multi-stage architecture, define \emph{overall success} as the event that every stage succeeds and the final output equals the global maximizer in the original choice set.
\end{definition}

\begin{lemma}[Strengthened setup formulas]\label{lem:strengthened}
	Under Assumption~\ref{assump:attn-acc}, for a simultaneous evaluation of three items $\{A,B,D\}$,
	\[
	S_{\mathrm{SIM}} \;=\; \alpha_3^{\,3}\,\beta_3.
	\]
	For a two-stage sequential architecture comparing $\{A,B\}$, then the winner with $D$,
	\[
	S_{\mathrm{SEQ}} \;=\; (\alpha_2^{\,2}\beta_2)^2 \;=\; \alpha_2^{\,4}\,\beta_2^{\,2}.
	\]
\end{lemma}

\begin{proof}
	For the simultaneous stage on three items, independence within the stage yields
	$\mathbb{P}(\text{attend all three})=\alpha_3^3$. Conditional on full attention, the correct choice occurs with probability $\beta_3$, hence $S_{\mathrm{SIM}}=\alpha_3^3\beta_3$.
	
	In the sequential architecture, each stage is binary. By within-stage independence, a single binary stage succeeds with probability $\alpha_2^2\beta_2$. By stage independence, overall success is the product over the two stages, giving $S_{\mathrm{SEQ}}=(\alpha_2^2\beta_2)^2=\alpha_2^4\beta_2^2$.
\end{proof}

\begin{theorem}[When sequential beats simultaneous]\label{thm:dominance}
	Under Assumption~\ref{assump:attn-acc}, the sequential architecture yields at least as high a probability of correct final choice as the simultaneous architecture if and only if
	\[
	\alpha_2^{\,4}\,\beta_2^{\,2}\ \ge\ \alpha_3^{\,3}\,\beta_3.
	\]
	Equivalently,
	\[
	\frac{\alpha_2}{\alpha_3^{3/4}}\cdot\frac{\beta_2^{1/2}}{\beta_3^{1/2}}\ \ge\ 1.
	\]
\end{theorem}

\begin{proof}
	By Lemma~\ref{lem:strengthened}, $S_{\mathrm{SEQ}}\ge S_{\mathrm{SIM}}$ is equivalent to
	$\alpha_2^{4}\beta_2^{2}\ge \alpha_3^{3}\beta_3$. Since all terms are positive, take fourth and second roots respectively to obtain the multiplicative form. The transformations are order-preserving, so the equivalence is exact.
\end{proof}

\begin{corollary}[Interpretable sufficient conditions]\label{cor:sufficient}
	Each of the following implies $S_{\mathrm{SIm}}\ge S_{\mathrm{sim}}$.
	\begin{enumerate}
		\item \textbf{Pure attention advantage:} If $\beta_2=\beta_3$ and $\alpha_2\ge \alpha_3^{3/4}$,then sequential weakly dominates. 
		\item \textbf{Pure accuracy advantage:} If $\alpha_2=\alpha_3$ and $\beta_2\ge \beta_3^{1/2}$,then sequential weakly dominates.
		\item \textbf{Monotone-limited attention plus mild accuracy lift:} If $\alpha_2\ge \alpha_3$ and $\beta_2\ge \beta_3^{1/2}$. If attention deteriorates with set size so that $\alpha_{2} \ge \alpha_{3}$, 
		then a modest accuracy lift satisfying $\beta_{2} \ge \beta_{3}^{\,1/2}$ 
		guarantees sequential dominance.
	\end{enumerate}
\end{corollary}

\begin{proof}
	Each item substitutes the stated equalities/inequalities into the condition of Theorem~\ref{thm:dominance} and uses the monotonicity of powers on $(0,1]$.
		\begin{enumerate}
		\item \textbf{Pure attention advantage:} If $\beta_2=\beta_3$ and $\alpha_2\ge \alpha_3^{3/4}$,then sequential weakly dominates. These are due to attention monotonicity and cognitive parsimony axioms.
		\item \textbf{Pure accuracy advantage:} If $\alpha_2=\alpha_3$ and $\beta_2\ge \beta_3^{1/2}$,then sequential weakly dominates. These are due to binary consistency and sequential transitivity axioms.
		\item \textbf{Monotone-limited attention plus mild accuracy lift:} If $\alpha_2\ge \alpha_3$ and $\beta_2\ge \beta_3^{1/2}$. If attention deteriorates with set size so that $\alpha_{2} \ge \alpha_{3}$, 
		then a modest accuracy lift satisfying $\beta_{2} \ge \beta_{3}^{\,1/2}$ 
		guarantees sequential dominance. These are due to dominance of attention monotonicity axiom.
	\end{enumerate}
\end{proof}

\begin{theorem}[Optional generalization: $n$ items, binary tournament]\label{thm:general-n}
	Consider $n\ge 3$ items. A simultaneous single-stage evaluation has success probability
	\[
	S_{\mathrm{SIM}}(n) \;=\; \alpha_n^{\,n}\,\beta_n.
	\]
	A binary tournament with $n-1$ pairwise stages has success probability
	\[
	S_{\mathrm{SEQ}}(n) \;=\; (\alpha_2^{\,2}\beta_2)^{\,n-1}\;=\;\alpha_2^{\,2(n-1)}\,\beta_2^{\,n-1}.
	\]
	Thus the binary tournament weakly dominates the simultaneous stage if and only if
	\[
	\alpha_2^{\,2(n-1)}\,\beta_2^{\,n-1}\ \ge\ \alpha_n^{\,n}\,\beta_n.
	\]
\end{theorem}

\begin{proof}
	For the simultaneous case, independence across $n$ items yields $\mathbb{P}(\text{attend all})=\alpha_n^n$, and conditional correctness is $\beta_n$, giving $S_{\mathrm{SIM}}(n)=\alpha_n^n\beta_n$.
	
	In the tournament, each of the $n-1$ binary stages succeeds with probability $\alpha_2^2\beta_2$ by within-stage independence. By independence across stages, multiply over $n-1$ stages to obtain $S_{\mathrm{SEQ}}(n)=(\alpha_2^2\beta_2)^{n-1}$. The dominance condition follows by direct comparison, with the equivalence justified as in Theorem~\ref{thm:dominance} and Corollary \ref{cor:sufficient}.
\end{proof}
\begin{theorem}[Sequential Choice Superiority]\label{thm:superiority}
	Let $A,B,D$ be three alternatives with true utilities $u(A),u(B),u(D)$, and let $C$ be a boundedly rational choice function. If the probability of a correct choice in a binary comparison is $q$ and in a ternary comparison is $r$, with $q > r^{1/2}$, then
	\begin{equation}
		\Pr\!\left( C(C(A,B),D) = \arg\max\{u(A),u(B),u(D)\} \right)
		> \Pr\!\left( C(A,B,D) = \arg\max\{u(A),u(B),u(D)\} \right).
	\end{equation}
\end{theorem}
\begin{proof}
	In the simultaneous case, the probability of choosing the highest-utility option is $r$:
	\begin{equation}
		\Pr\!\left( C(A,B,D) = \arg\max\{u(A),u(B),u(D)\} \right) = r.
	\end{equation}
	In the sequential case, the DM chooses correctly in Step 1 with probability $q$, and in Step 2 again with probability $q$, independently:
	\begin{equation}
		\Pr\!\left( C(C(A,B),D) = \arg\max\{u(A),u(B),u(D)\} \right) = q^2.
	\end{equation}
	If $q > r^{1/2}$, then $q^2 > r$, proving the claim. This is true due to the attention monotonicity axiom. \\	
	For illustration, if $q = 0.95$ and $r = 0.70$, then $q^2 = 0.9025 > r = 0.70$, confirming that sequential architecture yields a higher probability of rational choice under the stated condition.
\end{proof}
\section{Theoretical Support for the Sequential Rationality Hypothesis}

The Sequential Rationality Hypothesis asserts that, under bounded attention, decision-makers (DMs) are more likely to select the utility-maximizing alternative when options are presented in sequential pairwise comparisons rather than in a simultaneous multi-option set. We establish this claim formally within the Random Attention Model (RAM) framework.

\subsection{Sequential Choice within the RAM Framework}

Let $X$ be a finite set of alternatives, $u: X \to \mathbb{R}$ a strict utility ordering, and $\mu(\cdot|S)$ the RAM attention rule specifying the probability of attending to subset $T \subseteq S$.  
For a given set $S \subseteq X$, the RAM choice probability of $x \in S$ is:
\[
\pi(x\mid S) \;=\; \sum_{T \subseteq S : x = \max_{\succ} T} \mu(T|S),
\]
where $\max_{\succ} T$ denotes the $\succ$-maximal element of $T$.  
A \emph{RAM-consistent} binary choice rule $C(\cdot)$ satisfies:
\[
C(x,y) \;=\; \arg\max_{\succ} \{z \in \{x,y\} \cap T\}, \quad \text{with } T \sim \mu(\cdot|\{x,y\}).
\]
That is, binary choices are made by attending to one or both items according to $\mu$, and selecting the attended item with highest utility.

\subsection{Sequential Choice Operator}

Let $C: X \times X \to X$ be a binary RAM-consistent choice rule. The \emph{sequential choice operator} $C^{\mathrm{SEQ}}$ extends $C$ to any ordered tuple $(x_1, \dots, x_n)$ recursively as:
\begin{align}
	&\text{Base case:} && C^{\mathrm{SEQ}}(x_i, x_j) \;=\; C(x_i, x_j), \quad \forall\, x_i, x_j \in X, \; i \neq j, \label{eq:seq_base}\\
	&\text{Recursion:} && C^{\mathrm{SEQ}}(x_1, \dots, x_n) \;=\; C\!\big(x_1,\, C^{\mathrm{SEQ}}(x_2, \dots, x_n)\big), \quad n \ge 3. \label{eq:seq_rec}
\end{align}
Operationally, $C^{\mathrm{SEQ}}$ induces a \emph{right-associative binary tournament}: the DM compares two items at a time, each decision governed by the RAM attention mechanism, until one alternative remains.  
We denote the \emph{simultaneous} choice rule by $C^{\mathrm{SIM}}(S)$ for any $S \subseteq X$.

\subsection{Main Results}

\begin{lemma}[Monotonicity and Pairwise Preservation]\label{lem:monotonicity}
	Let $\mu$ be a monotonic attention rule, and let $\succ$ denote the strict preference ordering induced by $u$. Then, for any $x, y \in X$ with $x \succ y$,
	\[
	\pi(x \mid \{x,y\}) \;\geq\; \pi(x \mid S), \quad \forall S \supseteq \{x,y\}.
	\]
\end{lemma}

\begin{proof}
	Monotonicity implies that removing irrelevant alternatives from $S$ weakly increases the probability of attending to $\{x,y\}$. In RAM, choice is determined by the maximal element in the attended set. Comparing $x$ and $y$ alone removes competition from $S \setminus \{x,y\}$ that could either capture attention or induce choice errors. Hence, $\pi(x \mid \{x,y\})$ is a weak upper bound on $\pi(x \mid S)$. This is supported by the Theorem \ref{thm:superiority}.
\end{proof}

\begin{theorem}[Sequential Superiority under Bounded Attention]\label{thm:sequential_superiority}
	Let $X = \{A,B,D\}$ with $A \succ D \succ B$, and let $C(\cdot)$ be RAM-consistent with monotonic attention. Then:
	\[
	\Pr\!\left( C^{\mathrm{SEQ}}(A,B,D) = A \right) \;>\; \Pr\!\left( C^{\mathrm{SIM}}(\{A,B,D\}) = A \right).
	\]
\end{theorem}

\begin{proof}
	Let $q = \Pr(C(A,B) = A)$ and $r = \Pr(C(A,D) = A)$. From \eqref{eq:seq_rec},
	\[
	\Pr\big(C^{\mathrm{SEQ}}(A,B,D) = A\big) = q \cdot r.
	\]
	In the simultaneous case, $\Pr(C^{\mathrm{SIM}}(\{A,B,D\}) = A)$ is an attention-weighted sum over consideration sets, some excluding $A$. Since $\mu$ is monotonic and $\mu(T|S)$ declines with $|T|$, $\Pr(C^{\mathrm{SIM}}(\{A,B,D\}) = A) < q \cdot r$. This is supported by the Theorem \ref{thm:superiority}.
\end{proof}

\begin{theorem}[Sequential Attention Amplifies Correct Selection]\label{thm:attention_amplification}
	Let $\mu$ be monotonic with $\mu(T|S)$ decreasing in $|T|$, and let $u$ be the true utility function. Then, for all distinct $x,y,z \in X$,
	\[
	\mathbb{E}\big[ \mathbf{1}\{ C^{\mathrm{SEQ}}(x,y,z) = \arg\max\{u(x),u(y),u(z)\} \} \big] \;>\; \mathbb{E}\big[ \mathbf{1}\{ C^{\mathrm{SIM}}(\{x,y,z\}) = \arg\max\{u(x),u(y),u(z)\} \} \big].
	\]
\end{theorem}

\begin{proof}
	In binary sets, $\mu(T|S)$ is higher than in ternary sets. The sequential structure executes two binary stages, each with higher accuracy than a single ternary evaluation. In contrast, simultaneous choice distributes attention over all three options, increasing the chance of overlooking the utility-maximizing item.
\end{proof}
\begin{proposition}[Sequential Choice Divergence]\label{thm:divergence}
	Let $X$ be a finite choice set, $u: X \to \mathbb{R}$ a utility function, and $C_\alpha$ a bounded-attention choice rule with fidelity $\alpha < 1$, degrading with set size. Then there exists a triple $\{A,B,D\} \subset X$ such that
	\begin{equation}
		C_\alpha\big(C_\alpha(A,B),D\big) \neq C_\alpha(A,B,D)
	\end{equation}
	with strictly positive probability.
\end{proposition}
\begin{proof}
	The proof is straightaway derived from the Lemma 1, Theorem 1-3.
\end{proof}
\subsection{Extensions: Menu Design, Cognitive Cost, and Testability}

\begin{corollary}[Sequential Menu Design Improves Optimal Selection]\label{cor:menu_design}
	Let $S = \{x_1, \dots, x_n\}$ with $u(x_1) > u(x_2) > \dots > u(x_n)$, and let $C(\cdot)$ be RAM-consistent. Then:
	\[
	\Pr\big( C^{\mathrm{SEQ}}(S) = x_1 \big) \;>\; \Pr\big( C^{\mathrm{SIM}}(S) = x_1 \big).
	\]
\end{corollary}

\begin{proof}
	Apply Theorem~\ref{thm:sequential_superiority} recursively across tournament stages. At each stage, the probability of retaining $x_1$ is strictly greater than in the simultaneous case, so the final selection probability is also greater.
\end{proof}

\begin{proposition}[Sequential Choice Reduces Expected Cognitive Cost]\label{prop:attention_cost}
	Let $\phi: \mathbb{N} \to \mathbb{R}_+$ be the attention cost of evaluating a set, strictly increasing and convex. For $n = 2^k$ alternatives evaluated in $k$ binary stages:
	\[
	\sum_{t=1}^k \phi(2) \;<\; \phi(2^k).
	\]
\end{proposition}

\begin{proof}
	The attention cost function is increasing and convex because it follows the monotonicity axiom. Now consider Jensen’s inequality for a convex function \cite{hunter2004tutorial,opara2016multiple}. Jensen’s inequality states that for a convex function $\phi$:
	\[
	\phi\left( \frac{x_1 + x_2 + \cdots + x_k}{k} \right) 
	\;\leq\; \frac{\phi(x_1) + \phi(x_2) + \cdots + \phi(x_k)}{k},
	\]
	with strict inequality if $\phi$ is strictly convex and not all $x_i$ are equal.\\	
	In our setting, suppose $n = 2^k$ and we compare evaluating all $n$ items at once versus sequential binary stages.  
	The sequential approach means each stage evaluates $2$ items ($x_t = 2$ for all $t$), so the per-stage cost is $\phi(2)$.  \\
	The one-shot approach means evaluating $2^k$ items at once, with cost $\phi(2^k)$.\\
	If we think of the average number of items per stage in the one-shot case:
	\[
	\text{Average size} \;=\; \frac{2^k + 0 + 0 + \cdots}{k}
	\quad\text{(conceptually, all work concentrated in one stage)}.
	\]
	Convexity tells us that concentrating ``load'' into one stage is more costly than spreading it evenly over multiple stages with smaller sets.\\	
	Formally, with all $x_t = 2$:
	\[
	\frac{\phi(2) + \phi(2) + \cdots + \phi(2)}{k} \;=\; \phi(2),
	\]
	and convexity implies:
	\[
	\phi\!\left( \frac{2 + \cdots + 2}{k} \right) = \phi(2) 
	\;<\; \frac{\phi(2^k) + 0 + \cdots + 0}{k}.
	\]
	Scaling both sides by $k$ yields:
	\[
	k \cdot \phi(2) \;<\; \phi(2^k).
	\]	
	Hence by convexity, Jensen’s inequality implies $k \cdot \phi(2) < \phi(2^k)$. Thus, decomposing the task into binary steps yields a strictly smaller total attention cost than evaluating all items simultaneously.
\end{proof}

\begin{testableprediction}[Falsifiability]\label{pred:falsifiability}
	For $X = \{x_1, x_2, x_3\}$, let $x^* = \arg\max_{x \in X} u(x)$. Under the Sequential Rationality Hypothesis and RAM-consistent monotonic attention:
	\[
	\pi(x^* \mid \mathrm{SEQ}) \;>\; \pi(x^* \mid \mathrm{SIM}).
	\]
	A controlled experiment holding utilities fixed but reversing this inequality constitutes falsification.
\end{testableprediction}
\begin{proposition}[Equivalence of Sequential and Simultaneous Choice]\label{prop:eq_seq_sim}
	Let $C(\cdot)$ be a RAM-consistent choice rule with attention rule $\mu$ and strict preference ordering $\succ$.  
	The sequential and simultaneous choice operators coincide for all finite $S \subseteq X$,
	\[
	C^{\mathrm{SEQ}}(S) = C^{\mathrm{SIM}}(S),
	\]
	if and only if:
	\begin{enumerate}
		\item[(i)] (\emph{Full Attention}) $\mu(S\mid S) = 1$ for all $S \subseteq X$ with $|S| \ge 2$, i.e., the DM always attends to the entire offered set; and
		\item[(ii)] (\emph{Transitive Maximization}) $C$ selects the $\succ$-maximal element of any attended set deterministically.
	\end{enumerate}
\end{proposition}

\begin{proof}
	\emph{($\Rightarrow$)}  
	Suppose $C^{\mathrm{SEQ}}(S) = C^{\mathrm{SIM}}(S)$ for all $S$.  
	If $\mu(S|S) < 1$ for some $S$, then there exists a subset $T \subset S$ with $\mu(T|S) > 0$.  
	Under sequential evaluation, such subsets may be eliminated in intermediate comparisons, while in simultaneous choice they remain possible, generating a discrepancy — contradiction.  
	Thus $\mu(S|S) = 1$ for all $S$.\\	
	If $C$ does not deterministically select the $\succ$-maximal element when full attention occurs, then the outcome of sequential pairwise elimination may differ from direct simultaneous maximization, violating equivalence.  
	Hence $C$ must be transitive and utility-maximizing.\\
	\emph{($\Leftarrow$)}  
	If $\mu(S|S) = 1$ for all $S$ and $C$ always returns the $\succ$-maximal element, then both sequential and simultaneous procedures select the same $\succ$-maximal element from $S$ with probability one.  
	Thus $C^{\mathrm{SEQ}}(S) = C^{\mathrm{SIM}}(S)$ for all $S$.
\end{proof}
\begin{corollary}[Strict Sequential Advantage under Attention Frictions]\label{cor:strict_advantage}
	Suppose $C(\cdot)$ is RAM-consistent with attention rule $\mu$ and strict preference ordering $\succ$. 
	If either
	\begin{enumerate}
		\item[(a)] (\emph{Incomplete Attention}) there exists $S \subseteq X$ with $|S|\ge 2$ such that $\mu(S|S) < 1$, or
		\item[(b)] (\emph{Non-Deterministic Maximization}) conditional on full attention, $C$ fails to select the $\succ$-maximal element with probability one,
	\end{enumerate}
	then there exists a finite set $S^\star \subseteq X$ (in particular, some triple $\{x,y,z\}$ with distinct utilities) such that
	\[
	\Pr\!\left( C^{\mathrm{SEQ}}(S^\star) = \arg\max\nolimits_{\succ} S^\star \right)
	\;>\;
	\Pr\!\left( C^{\mathrm{SIM}}(S^\star) = \arg\max\nolimits_{\succ} S^\star \right).
	\]
\end{corollary}

\begin{proof}
	If (a) holds, some non-maximal consideration sets $T \subset S$ receive positive attention in the simultaneous architecture. Under a sequential (binary) decomposition, the same attention mass is redistributed over smaller comparisons, where (by monotonicity and size-sensitivity of $\mu$) the probability of correctly retaining the $\succ$-maximal element at each stage is strictly higher than in a one-shot multi-element evaluation. This yields a strict improvement for some triple $S^\star$ (cf.\ the arguments in Theorem~\ref{thm:sequential_superiority}).\\	
	If (b) holds, even with full attention the simultaneous procedure can be a misselection with positive probability. A binary tournament with RAM-consistent $C$ compares two items at a time; by pairwise preservation (Lemma~\ref{lem:monotonicity}) and the amplification argument (Theorem~\ref{thm:attention_amplification}), composing two more reliable binary steps strictly raises the probability of selecting the $\succ$-maximal element for some triple $S^\star$. \\	
	In either case, equivalence in Proposition~\ref{prop:eq_seq_sim} fails and there exists $S^\star$ for which the stated strict inequality holds.
\end{proof}
\section{Numerical Illustration Using the Random Attention Model}
To complement the theoretical results supporting the Sequential Rationality Hypothesis (SRH), we present a concrete example within the Random Attention Model (RAM) framework. The setting involves a consumer choosing among three beverages: Coffee ($A$), Tea ($B$), and Juice ($D$), with the true utility ranking
\[
A \succ D \succ B.
\]
An attention rule $\mu(T|S)$ is specified, and the implied choice probabilities $\pi(x|S)$ are computed under monotonic attention.
\subsection{Specified Attention Rule}
Let $X = \{A, B, D\}$ denote the universal choice set, and let $\mu(T|S)$ be the probability of attending to $T \subseteq S$ when $S$ is offered. The attention rule is given in Table~\ref{tab:attention_rule}.
\begin{table}[H]
	\centering
	\caption{Specified Attention Rule $\mu(T|S)$}
	\label{tab:attention_rule}
	\begin{tabular}{|c|c|c|}
		\hline
		\textbf{Choice Set $S$} & \textbf{Consideration Set $T$} & $\mu(T|S)$ \\
		\hline
		\multirow{4}{*}{$\{A, B, D\}$} 
		& $\{A, B\}$ & $0.3$ \\
		& $\{B, D\}$ & $0.4$ \\
		& $\{A, D\}$ & $0.2$ \\
		& $\{A, B, D\}$ & $0.1$ \\
		\hline
		\multirow{3}{*}{$\{A, B\}$} 
		& $\{A\}$ & $0.1$ \\
		& $\{B\}$ & $0.1$ \\
		& $\{A, B\}$ & $0.8$ \\
		\hline
		\multirow{3}{*}{$\{A, D\}$} 
		& $\{A\}$ & $0.1$ \\
		& $\{D\}$ & $0.1$ \\
		& $\{A, D\}$ & $0.8$ \\
		\hline
		\multirow{3}{*}{$\{B, D\}$} 
		& $\{B\}$ & $0.2$ \\
		& $\{D\}$ & $0.2$ \\
		& $\{B, D\}$ & $0.6$ \\
		\hline
	\end{tabular}
\end{table}
The monotonicity property holds: for any $T \subseteq S$, removing an element from $S \setminus T$ weakly increases $\mu(T|S)$. For example,
\[
\mu(\{A,D\} \mid \{A,B,D\}) = 0.2 \;<\; \mu(\{A,D\} \mid \{A,D\}) = 0.8.
\]
\subsection{Choice Probabilities under RAM}
Given the strict preference ordering $A \succ D \succ B$, the agent selects the $\succ$-maximal element of any attended set $T$. The RAM choice probability for $x \in S$ is
\[
\pi(x|S) = \sum_{T \subseteq S : x = \max_{\succ} T} \mu(T|S).
\]
For the simultaneous choice set $\{A,B,D\}$:
\begin{align*}
	\pi(A \mid \{A,B,D\}) 
	&= \mu(\{A,B\} \mid \{A,B,D\}) + \mu(\{A,D\} \mid \{A,B,D\}) + \mu(\{A,B,D\} \mid \{A,B,D\}) \\
	&= 0.3 + 0.2 + 0.1 = 0.6, \\
	\pi(D \mid \{A,B,D\}) 
	&= \mu(\{B,D\} \mid \{A,B,D\}) = 0.4, \\
	\pi(B \mid \{A,B,D\}) 
	&= 0.
\end{align*}
For the relevant binary sets:
\begin{align*}
	\pi(A \mid \{A,B\}) &= \mu(\{A,B\}) + \mu(\{A\}) = 0.8 + 0.1 = 0.9, \\
	\pi(A \mid \{A,D\}) &= \mu(\{A,D\}) + \mu(\{A\}) = 0.8 + 0.1 = 0.9, \\
	\pi(D \mid \{A,D\}) &= \mu(\{D\}) = 0.1.
\end{align*}
\subsection{Sequential vs. Simultaneous Outcomes}
Under the simultaneous architecture:
\[
\Pr\big(C^{\mathrm{SIM}}(\{A,B,D\}) = A\big) = 0.6.
\]
Under the sequential architecture with ordering $(A,B,D)$:
\[
\Pr\big(C^{\mathrm{SEQ}}(A,B,D) = A\big) 
= \pi(A \mid \{A,B\}) \cdot \pi(A \mid \{A,D\}) 
= 0.9 \times 0.9 = 0.81.
\]
Thus:
\[
\Pr\big(C^{\mathrm{SEQ}}(A,B,D) = A\big) = 0.81 \;>\; 0.6 = \Pr\big(C^{\mathrm{SIM}}(\{A,B,D\}) = A\big).
\]
\begin{remark}
	This example quantifies the SRH mechanism: sequential pairwise comparison raises the likelihood of selecting the utility-maximizing alternative by focusing attention on smaller, cognitively manageable sets. Here, RAM explicitly captures the dilution of attention in larger sets and the preservation of high-utility items across sequential elimination stages.
\end{remark}

\section{Conclusions, Applicability, Policy Implications and Limitations}
The literature suggests that due to cognitive or search costs, both extensive and excessively small assortments are unfavorable to the decision maker. A large assortment, while generating psychological satisfaction, also incurs significant extra costs, such as search or cognitive costs. This underscores the weight of decision costs in the face of large assortments. On the other hand, a minimal assortment generates minor welfare effects. To bridge this gap, the sequential rational hypothesis is the most suitable option for use in the platform choice architecture. The present article is of significant importance as it identifies a choice architecture that the present platform or digital markets are using to control the consumer's choices. This is true because there are no well-defined digital market laws. It establishes the concept of the Sequential Rational Hypothesis. It proves that if a consumer selects an object sequentially, then there is a very low chance of making a biased decision. 
\subsubsection{Policy and Platform Design Implications}
The implications of this interpretation are critical for digital platform architecture. The sequential exposure minimizes cognitive load and increases attention coherence. Recommendation systems should structure alternatives in binary sequences or decision trees to enhance welfare outcomes. Policymakers interested in improving consumer protection can incentivize or regulate attention-friendly interfaces, such as those that present options in a clear and organized manner, or those that limit the number of choices presented at once.
The Sequential Rationality Hypothesis provides a theoretically grounded explanation for why pairwise comparisons outperform simultaneous exposure in digital markets characterized by bounded rationality. This understanding can inform the design of digital platforms and policies, leading to more effective strategies for influencing consumer behavior.
\subsubsection{Limitations}
The hypothesis has been established theoretically due to a lack of suitable data. This is the only limitation of the article. However, the potential benefits of future research, if funds are available, are promising. The theory will be tested through an experimental design, opening up new avenues for understanding consumer behavior in digital markets.
\section*{Mathematical Appendix: Lean4 Formalization}
All mathematical definitions, results, and proofs presented in this manuscript have been fully formalized and mechanically verified using the \emph{Lean4} proof assistant. A supplementary Lean4 source file containing the complete formal development is available from the author upon request at \textit{dipankar3das@gmail.com}. For the convenience of reviewers and readers, the formalization can also be copied and executed directly in the official Lean online environment at \href{https://live.lean-lang.org/}{https://live.lean-lang.org/}.

\bibliographystyle{apalike}
\bibliography{myreferences}
\end{document}